%% file: main.tex
\documentclass[a4paper,UKenglish,svgnames]{lipics-v2019}

\usepackage{eulervm}
\usepackage{charter}
\usepackage{tikz}
\makeatletter
\let\c@author\relax
\makeatother
\usepackage{biblatex}
\usepackage{csquotes}
\usepackage{longtable}
\usepackage{amsmath}
\nocite{*}
\nolinenumbers

\graphicspath{{./images/}}
\newcommand{\pvrec}{\textsc{PV-Rec}}
\newcommand{\pdrec}{\textsc{PD-Rec}}
\newcommand{\pvman}{\textsc{PV-Man}}
\newcommand{\pdman}{\textsc{PD-Man}}
\newcommand{\domset}{\textsc{Dominating Set}}
\newcommand{\mulcliq}{\textsc{Multi-Colored Clique}}
\newcommand{\bd}{\textsc{$B_\mathcal{D}$}}
\newcommand{\ba}{\textsc{$B_\mathcal{A}$}}

\hideLIPIcs{}
\author{Kishen N. Gowda}{Indian Institute of Technology, Gandhinagar
}{kishen.gowda@iitgn.ac.in}{}{}
\author{Neeldhara Misra}{Indian Institute of Technology, Gandhinagar
}{neeldhara.m@iitgn.ac.in}{}{}
\author{Vraj Patel}{Indian Institute of Technology, Gandhinagar
}{vraj.patel@iitgn.ac.in}{}{}
\authorrunning{K. Gowda et al.}
\title{A Parameterized Perspective on Attacking and Defending Elections}
\begin{document}

\maketitle
\begin{abstract}
    We consider the problem of protecting and manipulating elections by recounting and changing ballots, respectively. Our setting involves a plurality-based election held across multiple districts, and the problem formulations are based on the model proposed recently by~[Elkind et al, IJCAI 2019]. It turns out that both of the manipulation and protection problems are NP-complete even in fairly simple settings. We study these problems from a parameterized perspective with the goal of establishing a more detailed complexity landscape. The parameters we consider include the number of voters, and the budgets of the attacker and the defender. While we observe fixed-parameter tractability when parameterizing by number of voters, our main contribution is a demonstration of parameterized hardness when working with the budgets of the attacker and the defender. 
    \end{abstract}
    \section{Introduction}
    \input{intro.tex}
    
    \section{Preliminaries}

\input{prelims.tex}
    
    \section{Plurality over Voters (PV)}
    \input{pv.tex}

    \section{Plurality over Districts (PD)}
    \input{pd.tex}
    
    \section{Concluding Remarks}
    \input{concl.tex}

\section*{References}
\printbibliography[heading=none]
\end{document}

%% file: intro.tex
Electoral fraud and errors in consolidation of large-scale voting data are fundamental issues in democratic societies. To counteract issues with malicious manipulations and accidental errors in the counting of votes, most electoral systems allow for strategic recounting of ballots to verify the election outcome. Recounting is generally an expensive and high-stakes process, and it would be desirable to formalize the problem so as to capture the relevant trade-offs and possibly pursue an algorithmic approach to finding an optimal recounting strategy. Such a framework was recently proposed by Elkind et al.~\cite{ElkindGORV19}, where the authors considered the problems of protecting and manipulating elections by recounting and changing ballots, respectively. These problems are modeled as a Stackelberg game involving an attacker and a defender. Both players work with limited budgets (say $\mathcal{B}_A$ and $\mathcal{B}_D$), and the question is if the players can develop optimal strategies for their desired outcomes. 

In this model, the election is spread out across multiple districts, with the voter preferences aggregated according to the plurality voting rule in one of two different ways, which we will explain explicitly in a moment. The manipulation problem is the following. The attacker has to optimize, typically with the goal of turning the election in favor of a particular candidate that he may have in mind, an attack strategy that involves manipulating the votes in at most $\mathcal{B}_A$ districts while ensuring that the impact of the attack persists \emph{even if} the defender restores at most $\mathcal{B}_D$ of these districts to their original state. In the recounting problem, the defender is given complete information about the original and manipulated voting profiles and she can restore the state of at most $\mathcal{B}_D$ districts with the goal of making the ``true winner'' win the repaired election. 

\paragraph*{Known Results.} The results obtained in~\cite{ElkindGORV19} already demonstrate the hardness of the attacker's and defender's problems for two natural ways of aggregating the votes: (1) \emph{Plurality over Voters} (PV), where districts are only used for the purpose of collecting the ballots and the winner is selected among the candidates that receive the largest number of votes in total, and (2) \emph{Plurality over Districts} (PD), where each district selects a preferred candidate using the Plurality rule, and the overall winner is chosen among the candidates supported by the largest number of districts, or the set of districts with largest total weight if the districts have weights associated with them. We briefly recall the main highlights from~\cite{ElkindGORV19}, since this provides the context for our contributions. It turns out that the recounting problem is NP-complete for both implementations of the plurality rule, even when there are only three candidates (this result assumes a succinct representation of the votes) or even when the votes are specified in unary. The problem is tractable when PD is employed over unweighted districts. On the other hand, the manipulation problem is NP-hard for PD even with unweighted districts, and in fact, $\Sigma_2^P$-complete for PD with succinct input even when there are only three candidates. Further, it is NP-hard for PV, again even when there are only three candidates (in the setting of succinct input) or even when the votes are specified in unary.


\paragraph*{Our Contributions.} Our main contribution is to establish the parameterized intractability of both the recounting and manipulation problems under both implementations of the plurality voting rule when parameterized by the budget of the players. Our contributions directly address a direction suggested by~\cite{ElkindGORV19}. In particular, we obtain the following results:

\begin{theorem}
The \pvrec{} and \pdrec{} problems are W[2]-hard and W[1]-hard, respectively, when parameterized by the budget of the defender, and are FPT when parameterized by the number of districts.
\end{theorem}

\begin{theorem}
The \pvman{} and \pdman{} problems are W[1]-hard when parameterized by the budget of the attacker (even when the defender budget is zero), and are FPT when parameterized by the number of voters.
\end{theorem}

Our results rely on reductions from traditional problems such as \mulcliq{} and \domset{}. Our hardness results work even when the input is specified in unary. It is reasonably natural to imagine that these parameters would be small in practice, since they correspond to real-world budget constraints. To that end, our results here bring mixed news: on the one hand, the hardness of mounting an attack may be viewed as a positive outcome, but on the other hand, it turns out that the problem of optimally reversing damages is hard as well. This triggers the natural question of whether the recounting problem admits good approximations when treated as an optimization problem, either on the criteria of the budget or on the criteria of the quality of the winning candidate that we are able to restore. The FPT algorithms that we present rely mostly on straightforward enumeration, and it would be interesting to improve the running times in question. 

\paragraph*{Related Work.} While we build most closely on the work of Elkind et al~\cite{ElkindGORV19}, and much of the work on manipulation in the literature of social choice does not consider the possibility of a counter-attack, we note that some recent investigations have been carried out in a spirit that is similar to our present contribution. Dey et al.~\cite{ijcai2019-34} also consider a parameterized approach to protecting elections, where the voting rule in question is the Condorcet rule. They build on the work of Yin et al.~\cite{YinVAH16}, who study a pre-emptive approach to protecting elections. In these models, the defender allocates resources to guard some of the electoral districts, so that the votes there cannot be influenced, and the attacker responds afterward. This is in contrast to our setting, where the defender makes the second move. The social choice literature is rich in studies of manipluation, control, and bribery. For a detailed overview we direct the reader to the surveys~\cite{ConitzerW16,FaliszewskiR16}.

%% file: prelims.tex
We recall the setting from~\cite{ElkindGORV19}. We consider elections over a \emph{candidate} set $C$, where $|C| = m$. There are $n$ \emph{voters} who are partitioned into $k$ pairwise disjoint \emph{districts}. The set of all districts is $\mathcal{D}$. For each $i \in \mathcal{D}$, let $n_i = |i|$. We note that in this context, $i$ denotes a subset of voters. For each $i \in \mathcal{D}$, district $i$ has a \emph{weight} $w_i$, which is a positive integer. We say that an election is \emph{unweighted} if $w_i = 1$ for all $i \in \mathcal{D}$. Each voter votes for a single candidate in $C$. For each $i \in \mathcal{D}$ and each $c \in C$ let $\textbf{v}_{ic}$ denote the number of votes that candidate $c$ gets from voters in $i$. We refer to the list $\textbf{v} = (\textbf{v}_{ic})_{i \in \mathcal{D}, c \in C}$ as the \emph{vote profile}. 

Let $\succ$ be a linear order over $C$; $a \succ b$ indicates that $a$ is favored over $b$. We consider the following two voting rules, which take the vote profile $\textbf{v}$ as their input.
\begin{itemize}
    \item \emph{Plurality over Voters} (PV): We say that a candidate $a$ \emph{beats} a candidate $b$ under PV if $\sum\limits_{i \in \mathcal{D}}\textbf{v}_{ia} > \sum\limits_{i \in \mathcal{D}}\textbf{v}_{ib}$ or $\sum\limits_{i \in \mathcal{D}}\textbf{v}_{ia} = \sum\limits_{i \in \mathcal{D}}\textbf{v}_{ib}$ and $a \succ b$; the winner is the candidate that beats all other candidates. Note that district weights $w_i$ are not relevant for this rule.
    \item \emph{Plurality over Districts} (PD): For each $i \in \mathcal{D}$ the winner $c_i$ in $i$ is chosen from the set $arg\:max_{c \in C} \textbf{v}_{ic}$, with ties broken according to $\succ$. Then, for each $i \in \mathcal{D}$, $c \in C$, we set $w_{ic} = w_i$ if $c = c_i$, else $w_{ic} = 0$. We say that a candidate $a$ \emph{beats} a candidate $b$ under PD if $\sum\limits_{i \in \mathcal{D}} w_{ia} > \sum\limits_{i \in \mathcal{D}} w_{ib}$ or $\sum\limits_{i \in \mathcal{D}} w_{ia} = \sum\limits_{i \in \mathcal{D}} w_{ib}$ and $a \succ b$. The winner is the candidate that beats all other candidates. Given the voting profile $\textbf{v}$, we take the winner in district $i$ to be $\mathcal{G}_\textbf{v}(i)$. We shall omit the subscript if the voting profile to be used is clear from the context.
\end{itemize}

For PV and PD, we define the \emph{social welfare} of a candidate $c \in C$ as the total number of votes that $c$ gets and the total weight that $c$ gets, respectively:
$$SW^{PV}(c) = \sum\limits_{i \in \mathcal{D}} \textbf{v}_{ic}, \:\:\:SW^{PD}(c) = \sum\limits_{i \in \mathcal{D}} w_{ic}$$
Hence, the winner under each voting rule is a candidate with the maximum social welfare. We now define some additional terminology that we use later in this paper: 
\begin{itemize}
    \item \emph{Score:} Given a voting profile $u$, under the PD rule, the score of a candidate $p$ is defined as the sum of weights of districts in which the candidate $p$ wins. Formally, $sc_u(p) = \sum\limits_{i \in \mathcal{D}} w_i \cdot \delta_{\mathcal{G}(i)p}$ (Here $\delta$ is the Kronecker delta function).
    \item \emph{Rivals:} Given a voting profile $u$, under the PD rule, we define the rivals of a candidate $p$ to be the set of candidates $\mathcal{C} \subset C$ such that for all candidates $c \in \mathcal{C}$, either $sc_u(c) > sc_u(p)$ or $sc_u(c) = sc_u(p)$ and $c \succ p$. 
\end{itemize}

We now consider the scenario where an election may be manipulated by an attacker, who wants to change the result of the election in favor of his preferred candidate $p$. The attacker has a budget $B_{\mathcal{A}} \in \mathbb{N}$ which enables him to change the voting profiles in at most $B_{\mathcal{A}}$ districts. For each district $i \in \mathcal{D}$, we define $\gamma_i,0\leq\gamma_i\leq n_i$ to be the number of votes that the attacker can manipulate in $i$. After the manipulation we have a voting profile $\overline{\textbf{v}} = (\overline{\textbf{v}}_{ic})_{i \in \mathcal{D}, c \in C}$. We formalize the notion of a manipulation as a set $M \subseteq \mathcal{D}$ and a voting profile $\overline{\textbf{v}}$ such that $|M| \leq B_{\mathcal{A}}$ where $\overline{\textbf{v}}_{ic} = \textbf{v}_{ic}$ for all $i \notin M$ and for all $i\in M$ it holds that $\sum\limits_{c \in C}\overline{\textbf{v}}_{ic} = n_i$ and $\sum\limits_{c \in C} \frac{|\overline{\textbf{v}}_{ic} - \textbf{v}_{ic}|}{2} \leq \gamma_i$. 

After the attacker, a socially minded defender with budget $B_\mathcal{D} \in \{0\} \cup \mathbb{N}$ can demand a recount in at most $B_\mathcal{D}$ districts. Formally, a recounting strategy $R$ is a set such that $R \subseteq M$ and $|R| \leq B_\mathcal{D}$. After the defender recounts, the vote counts of the districts in $R$ are restored to their original values. This results in a new voting profile $u = (u_{ic})_{i \in \mathcal{D}, c \in C}$ where for all $i \in \mathcal{D}\setminus M \cup R, u_{ic} = \textbf{v}_{ic}$ and for all $i \in M \setminus R, u_{ic} = \overline{\textbf{v}}_{ic}$. Then the voting rule $\mathcal{V} = \{PV, PD\}$ is applied to the profile $u$ that is obtained to obtain a winner (let it be $w'$) with ties broken according to $\succ$. The defender's objective is to maximize $SW^{\mathcal{V}}(w')$. It is a game of perfect information i.e. both entities know all information about the game. 

We say that the attacker wins if he has a manipulation strategy such that after the defender moves optimally, the preferred candidate of attacker i.e. $p$ wins. We define the following two decision problems based on voting rule $\mathcal{V} \in \{PV, PD\}$ and the two entities:
\begin{itemize}
    \item $\mathcal{V}$-\textsc{Man}: Given the voting rule $\mathcal{V}$, a voting profile $\textbf{v}$, the linear order $\succ$, a preferred candidate $p$, attacker budget $B_\mathcal{A}$, defender budget $B_{D}$ and weights $w_i$ and parameter $\gamma_i$ for each district $i \in \mathcal{D}$, does the attacker have a winning strategy?
    \item $\mathcal{V}$-\textsc{Rec}: Given the voting rule $\mathcal{V}$, a voting profile $\textbf{v}$, a manipulated voting profile $\overline{\textbf{v}}$, a preferred candidate $w$, the linear order $\succ$, defender budget $B_{D}$ and weights $w_i$ for each district $i \in \mathcal{D}$, can the defender make $w$ win by recounting at most $B_\mathcal{D}$ districts? 
\end{itemize}

Next we state the definitions and parameterized hardness results for decision problems that are used throughout this paper, and refer the reader to~\cite{downey99parameterized,CyganEtAl} for a detailed introduction to parameterized complexity and the framework of parameterized reductions. 
\begin{itemize}
    \item \domset: A set of vertices $D$ is a dominating set in graph $G$ if $V (G) = N_G[D]$. \domset{} asks that given a graph $G$ and a non-negative integer $k$, does there exist a dominating set of size at most $k$? \domset{} is known to be W[2]-hard parameterized by $k$~\cite{downey99parameterized}.
    \item \mulcliq: Given a graph $G$ and a partition of the vertex set $V$ into $k$ color classes $V_1, V_2, \hdots V_k$, \mulcliq{} asks whether there exists a clique of size $k$ with one vertex each from $V_1, V_2, \hdots V_k$. \mulcliq{} is known to be W[1]-hard parameterized by $k$~\cite{FellowsHRV09,Pietrzak03}.
\end{itemize}

%% file: pv.tex
In this section, we analyze the parameterized complexity of \pvrec{} and \pvman{} with different parameters. It is easy to see that \pvrec{} is FPT when parameterized by the number of districts ($k$), since we may guess the districts to be recounted. Since $k \leq n$, the problem is also FPT parameterized by the number of voters. 

\begin{proposition}
\label{lemma1}
    \pvrec{} is FPT when parameterized by the no. of districts $k$ or the number of voters $n$.
\end{proposition}

We now show that \pvrec{} is W[2]-hard when parameterized by budget of the defender (\bd). Before describing the construction formally, we briefly outline the main idea. We reduce from the \domset, which is well-known to be W[2]-hard parameterized by the size of the solution, which we denote by $k$. Let $(G = (V,E), k)$ be an instance of \domset{}. We create an instance of \pvrec{} where we have candidates and districts corresponding to vertices of $G$, along with a special candidate $w$ who is our desired winner. To begin with, we have an ``immutable'' district --- one where the original and manipulated votes are identical --- that sets the baseline score of the special candidate at $n$. The number votes for any other candidate $c$ from this district is fixed to ensure that the total number of votes for $c$ is \emph{also} $n$. In a district corresponding to a vertex $v$, every candidate corresponding to a vertex $u \in N[v]$ gets one vote. In the original scenario, all voters in these districts vote only for some dummy candidates. The key is that a ``switch'' in a district corresponding to a vertex $v$ \emph{reduces} the vote count for all vertices in $N[v]$. Since $w$ receives no votes from any of the other districts, observe that the only way for $w$ to emerge as a unique winner is if all of the other candidates \emph{lose} votes from the switches. It is not hard to infer from here that the defender has a valid switching strategy if and only if $G$ has a dominating set of size at most $k$.  

\begin{lemma}
    \pvrec{} is W[2]-hard parameterized by \bd, the defender's budget.
\end{lemma}
\begin{proof}
    We present an FPT reduction from the \domset{} problem. Let $(G = (V,E), k)$ be an instance of \domset{}. Let $N = |V|$ and $M = |E|$. We begin by describing the construction of the reduced instance.

    
    \textbf{Districts:} We introduce a \emph{baseline} district $\mathcal{D}_0$. Further, for each vertex $v \in V$, we introduce a corresponding \emph{primary} district $\mathcal{D}_v$.
    
    \textbf{Candidates:} For each vertex $v \in V$ we introduce a \emph{main} candidate $c_v$ and a \emph{dummy} candidate $d_v$. We also have a \emph{special} candidate $w$. 
      
    \textbf{Voting Outcomes:} The voting outcomes are as follows. For ease of presentation, let $v \in V$ be arbitrary but fixed. 
    \begin{enumerate}
        \item The special candidate does not receive any votes from the primary districts in either the original or the manipulated settings. In particular, $\overline{\textbf{v}}_{\mathcal{D}_v,w} = \textbf{v}_{\mathcal{D}_v,w} = 0$.
        \item In the original election, the main candidates have no votes in the primary districts, that is, $\textbf{v}_{\mathcal{D}_v,c_u} = 0$ for all $u \in V$. 
        \item In the manipulated election, a main candidate $c_u$ has a single vote in its favor in the primary district $\mathcal{D}_v$ if and only if $u \in N_G[v]$. Formally,
        $$\overline{\textbf{v}}_{\mathcal{D}_v,c_u} = \begin{cases} 1 &\quad\text{if $u \in N_G[v]$}\\ 0 &\quad\text{otherwise}\\ \end{cases}$$
        \item In the original election, a dummy candidate $d_u$ has a score of $d(u)+1$ in the primary district corresponding to $u$, and a score of zero everywhere else. Formally,
        $$\textbf{v}_{\mathcal{D}_v,d_{u}} = \begin{cases} d(v)+1 &\quad\text{if $u = v$}\\0 &\quad\text{otherwise}\\ \end{cases}$$
        \item The dummy candidates receive no votes in the primary districts in the manipulated elections. 
        \item In the baseline district, the score of the main candidates is defined to ensure that their total score in the manipulated election is $N$. In particular, $\overline{\textbf{v}}_{\mathcal{D}_0,c_v} = \textbf{v}_{\mathcal{D}_0,c_v} = N - (d (v) + 1)$.
        \item The dummy candidates receive no votes in the baseline district in both the original and manipulated elections.
        \item The score of $w$ is $N$ in the baseline district in both the original and manipulated elections. In particular, $\overline{\textbf{v}}_{\mathcal{D}_0,w} = \textbf{v}_{\mathcal{D}_0,w} = N$.
    \end{enumerate}
    To summarize, the primary districts corresponding to a vertex $v$ have $d(v) + 1$ voters, and all main candidates corresponding to vertices in $N_G[v]$ get one vote each in the manipulated world; while the dummy candidate $d_v$ gets all the votes in the original world. Observe that in the manipulated election, all the candidates except the dummy candidates have a total score of $N$, while the dummy candidates have a score of zero. 
    
    We set \bd $= k$. The preferred candidate is $w$. We also work with the following tie-breaking order:
 $\hdots c_{v} \hdots \succ w \succ \hdots d_{vu} \hdots,$ where the main candidates are preferred over the special candidate, but the special candidate dominates the dummy candidates. This completes the description of the constructed instance. We now turn to the proof of equivalence. 

 \textbf{Forward Direction.} Let a dominating set $S$ of size at most $k$ be given. Now, select the districts $\mathcal{D}_v$ for all $v \in S$ to recount. After recounting, for every vertex $v$ in the dominating set, the votes of $c_u$ for all $u \in N_G[v]$ will decrease by one. Since $\bigcup\limits_{v \in S} N[v] = V$, all candidates corresponding to vertices will lose at least one vote each. Also, no dummy candidate can gain more than $N$ votes after recounting. So, $w$ has more votes than any main candidate and at least as many votes as any of the dummy candidates. Therefore, $w$ will win after recounting.
    
\textbf{Reverse Direction.} Conversely, suppose that the defender has a valid recounting strategy $\mathcal{Z}$ that results in making $w$ the winner. Since $|\mathcal{Z}|\le \bd$, so at most $k$ districts can be recounted. Observe that any optimal solution will not recount the baseline district, since it does not affect the outcome, so without loss of generality, every recounted district corresponds to a vertex $v \in V$. We claim that the vertices corresponding to the recounted districts, which we denote by $S$, constitute a dominating set in $G$.

We argue this by contradiction --- indeed, assume that $S$ is not a dominating set for $G$. Then there exists at least one vertex $u \in V$ such that $S \cap N[u] = \emptyset$. Then, observe that the score of the candidate $c_{u}$ remains $N$ after the recounting. In particular, $c_{u}$ has the same score as the special candidate $w$ after recounting, since the construction ensures that it is not possible to change the score of $w$ by recounting. Since the main candidates dominate the special candidate in the tie-breaking order, we have that $c_u$ wins the election, which contradicts our assumption that $\mathcal{Z}$ was a valid recounting strategy.
\end{proof}




We now turn our attention to \pvman{}. First, we prove that \pvman{} is FPT parameterized by the number of voters. This follows by first observing that $m \leq 2n$ without loss of generality, since at most $2n$ candidates can have a non-trivial score across the original and manipulated instances combined, and the remaining candidates are irrelevant to the instance. The algorithm can then proceed by guessing a manipulation strategy --- note that the space of all possible strategies is bounded once the candidates are bounded --- and then invoking the \pdrec{} algorithm from the previous section as a subroutine to verify the validity of the guessed strategy. Thus, we have the following.

\begin{proposition}
    \textsc{PV-Man} is FPT parameterized by $n$, the number of voters.
\end{proposition}

We now show that \pvman{} is W[1]-hard parameterized by budget of the attacker (\ba). Before describing the construction formally, we briefly outline the main idea. We reduce from the \mulcliq{} problem, which is well-known to be W[1]-hard parameterized by the number of color classes, which we denote by $k$. Let $(G = (V = V_1 \uplus \cdots \uplus V_k,E), k)$ be an instance of \textsc{Multicolored-Clique}. In the reduced instance, we introduce a special candidate $w$ who is the preferred candidate of the defender. The rival candidates are candidates corresponding to color classes $\mathcal{R}_i$, ordered pairs of color classes $\mathcal{R}_{i j}$ and vertices $c_v$. We also introduce some \emph{dummy} candidates. We introduce districts corresponding to each $v \in V$ and each $e\in E$. Also, there exists a special district which is ``immutable'', which sets up the initial scores of all candidates such that they are equal to a large number (say $F$). Initially, $w$ has $0$ votes. The scores are set up in such a way that the attacker has to transfer votes of $k^2$ districts to $w$ to make her win. The scores are engineered to ensure that the attacker has a successful manipulation strategy if and only if these $k^2$ districts correspond to $k$ vertices and $\binom{k}{2}$ edges that form a multicolored clique in $G$. 

\begin{lemma}
    \textsc{PV-Man} is W[1]-hard parameterized by \ba, the attacker's budget.
\end{lemma}
\begin{proof}
We demonstrate a parameterized reduction from the \mulcliq{} problem. Let $(G = (V = V_1 \uplus \cdots \uplus V_k,E), k)$ be an instance of \mulcliq{}. We begin by describing the construction.

\textbf{Districts:} There are two types of districts. We have a \emph{primary district} $\mathcal{D}_v$ for each vertex $v \in V$ and two \emph{secondary districts} $\mathcal{D}_{uv}$ and $\mathcal{D}_{vu}$ for each edge $e = (u,v) \in E$. Apart from these, there is a \emph{baseline district} $\mathcal{D}_0$.

\textbf{Candidates:} For each vertex $v \in V$ we will have a \emph{main} candidate $c_v$. Also, we have \emph{challenger} candidates corresponding to color classes $\mathcal{R}_i$'s and ordered pairs of color classes $\mathcal{R}_{ij}$'s. We introduce some \emph{dummy} candidates of an unspecified number, whose role is equalize the number of votes across primary and secondary districts. Finally, we have a \emph{special} candidate $w$.

\textbf{Voting profiles:} We introduce the following voting outcomes for the candidates. 

\begin{enumerate}
    \item The score of the special candidate $w$ is zero in all districts.
    \item A main candidate $c_v$ for $v \in V$ has a score of $k-1$ in the primary district corresponding to $v$, and a score of zero in all other primary districts. 

    \item A main candidate $c_v$ for $v \in V$ has a score of one in every secondary district corresponding to an edge $e = \{u,w\}$ that it is \emph{not} incident to, provided $u$ and $v$ share the same color class. In particular, if $v \in V_i$ and there are $t$ edges incident on $V_i \setminus \{v\}$, then $c_v$ has a score of one in $t$ secondary districts. Formally, assuming $v \in V_i$, we have: 
    
    $$\textbf{v}_{\mathcal{D}_{uw},c_v} = \begin{cases} 1
    &\quad\text{if $u \neq v$ and $u\in V_i$}\\ 0 &\quad\text{otherwise}\\ \end{cases}$$

    \item A challenger candidate corresponding to color class $i$ has one vote from any primary district corresponding to a vertex $v \in V_i$ and a score of zero from all other primary districts. In other words:
$$\textbf{v}_{\mathcal{D}_v,\mathcal{R}_i} = \begin{cases} 1 &\quad\text{if $v \in V_i$}\\ 0 &\quad\text{otherwise}\\ \end{cases}$$

    \item  A challenger candidate corresponding to an ordered pair of color classes $(V_i,V_j)$ has one vote from any secondary district $\mathcal{D}_{uv}$ corresponding to an edge $e \in E$ whose endpoints $u$ and $v$ are in color classes $V_i$ and $V_j$ respectively, and a score of zero from all other secondary districts. Note that the candidates $\mathcal{R}_{ij}$ and $\mathcal{R}_{ji}$ receive scores of one from distinct secondary districts. Specifically, we have: $\textbf{v}_{\mathcal{D}_{uv},\mathcal{R}_{i j}} = 1$ if $u \in V_i$ and $v \in V_j$.

    \end{enumerate}

Now, let $\ell$ be the size of the largest --- in terms of the number of voters --- among the primary and secondary districts constructed so far. For every primary or secondary district $D$ with $\nu(D)$ voters, we add $\ell - \nu(D)$ dummy voters and dummy candidates, and we let each dummy voter vote for a distinct dummy candidate. We let $F = \ell k^2$. 

We are now ready to specify the voting outcomes from the baseline district. these are simply designed to ensure that all primary candidates $c_v$ get $F + k-2$ votes and the challenger candidates get $F$ votes, which is easy to verify from the proposed outcomes below:

\begin{enumerate}
    \item $\textbf{v}_{\mathcal{D}_0,c_v} = F-1- \sum\limits_{e = \{u,v\} \in E} (\textbf{v}_{\mathcal{D}_{uv},c_v} + \textbf{v}_{\mathcal{D}_{vu},c_v}$)
    \item $\textbf{v}_{\mathcal{D}_0,\mathcal{R}_i} = F-\sum\limits_{u \in V}\textbf{v}_{\mathcal{D}_u,\mathcal{R}_i} - \sum\limits_{e = \{u,v\} \in E}(\textbf{v}_{\mathcal{D}_{uv},\mathcal{R}_i} + \textbf{v}_{\mathcal{D}_{vu},\mathcal{R}_i})$
    \item $\textbf{v}_{\mathcal{D}_0,\mathcal{R}_{i j}} = F-\sum\limits_{u \in V}\textbf{v}_{\mathcal{D}_u,\mathcal{R}_{i j}} - \sum\limits_{e = \{u,v\} \in E}(\textbf{v}_{\mathcal{D}_{uv},\mathcal{R}_{i j}} + \textbf{v}_{\mathcal{D}_{vu},\mathcal{R}_{i j}})$
\end{enumerate}

Note that apart from the above, the dummy candidates get $1$ vote each, and the special candidate $w$ has $0$ votes. Also, the attacker has no room to manipulate in the baseline district, that is, $\gamma_{\mathcal{D}_0} = 0$. On the other hand, the attacker can modify up to $\ell$ votes in the primary and secondary districts. Further, we set $\ba = k^2$ and $\bd = 0$. The preferred candidate is $w$. Finally, we impose the following tie-breaking order:

$$\hdots \mathcal{R}_i \hdots \succ \hdots \mathcal{R}_{i j} \hdots \succ \hdots c_{v} \hdots \succ w \succ \hdots \mbox{ dummies} \hdots.$$

This completes the description of the constructed instance. We now turn to the proof of equivalence. 

\textbf{Forward Direction.} Let a multi-colored clique $S$ of size $k$ be given. The attacker chooses the $k$ primary and $2 \binom{k}{2}$ secondary districts corresponding to the vertices and edges of $S$, and transfers \emph{all} the $\ell$ votes in these $k^2$ districts to the desired candidate $w$. The score of $w$ is now $F$. Further, note that that the scores of all challenger candidates has decreased by one to $F-1$, and the scores of all main candidates have decreased by $k-1$ as well, but for different reasons: for main candidates corresponding to vertices of the clique, the drop is directly from the recounting in the primary districts, while for any other main candidate, the drop is cumulative across $(k-1)$ relevant secondary districts. In particular, suppose $S \cap V_i := \{v_i\}$. Consider any $u \in V_i$ such that $u \neq v_i$, and observe that $c_u$ had a score of one in the following $(k-1)$ districts:

$$D_{v_i v_1}, \ldots, D_{v_i v_{i-1}},D_{v_i v_{i+1}}, \ldots, D_{v_i v_k},$$

which have indeed been attacked, and therefore the score of $c_u$ reduces by $k-1$. This leaves all candidates ranked ahead of the special candidate $w$ in the tie-breaking order with a score less than the final score of $w$, and the scores of the dummy candidates is either zero or one, thus they pose no threat to $w$. Therefore, $w$ wins the election under this attack, concluding the argument in the forward direction.


\textbf{Reverse Direction.}  Suppose that the attacker has a valid manipulation strategy $\mathcal{Z}$ that results in making $w$ the winner. With a budget of $k^2$ and manipulable districts with $\ell$ voters, recalling that $F = \ell k^2$, we observe that the maximum score that an attacker an achieve for the special candidate $w$ is $F$. Since challenger candidates have a score of $F$ and are ahead of $w$ in the tie-breaking order, we conclude that for each $i \in [k]$, $\mathcal{Z}$ must contain at least one primary district corresponding to a vertex from $V_i$ --- indeed, if not, the attack does not influence the score of challenger candidates corresponding to some color class, and attack will not result in a win for $w$. We claim that these vertices correspond to a clique in $G$. We argue this claim by contradiction. To begin with, we note that analogous to the vertex-based challenger candidates, the attack is forced to have a certain structure to account for the edge-based challeger candidates. In particular, for each ordered pair of color classes $(V_i,V_j)$, the attack must recount in a secondary district $\mathcal{D}_{pq}$ such that $p \in V_i$ and $q \in V_j$. Note that the existence of such a district implies that $(p,q) \in E$. Also, recalling the argument made earlier for the existence of primary districts corresponding to each color class in the attack, and combining this with the available budget of $k^2$, it is easy to see that \emph{any} successful attack has the following specific form: it involves exactly $k$ primary districts corresponding to $k$ vertices from $k$ different color classes, and $k(k-1)$ secondary districts corresponding to each ordered pair of color classes. 

Now, let the vertices derived from the attack be $\{v_1, \ldots, v_k\}$ where $v_i \in V_i$, and suppose, for the sake of contradiction, that $(v_i,v_j) \notin E(G)$. The challenger candidates $\mathcal{R}_{ij}$ and $\mathcal{R}_{ji}$ force the attacker to recount in two secondary districts corresponding to edges that have endpoints in $V_i$ and $V_j$. Suppose the secondary district recounted to account for $\mathcal{R}_{ij}$ was $D_{xy}$, with $x \in V_i$ and $y \in V_j$. Suppose, without loss of generality, that $x \neq v_i$. We know from the original score of $x$ and the maximum possible final score of $w$ that a successful attack needs to ensure that the score of $x$ reduces by at least $k-1$. However, note that this decrease cannot be addressed by a primary district in the attack $\mathcal{Z}$, since this is already ``used up'' for $v_i \neq x$. Also, among the remaining secondary districts chosen, it is easy to verify that there are at most $(k-2)$ that give a score of one to the candidate $x$, and therefore, we can only hope to reduce the score of $x$ by at most $k-2$ overall. Note that we are in this situation because $x$ has a score of zero in the district $\mathcal{D}_{xy}$, which might be intuitively viewed as a ``lost opportunity'' for reducing the score of $x$. This establishes that the attack in question is futile in it's effectiveness towards making $w$ win, clearly contrary to our assumption. Therefore, $\mathcal{Z}$ must choose primary districts which correspond to a multi-colored clique in $G$, and this completes the argument in the reverse direction.
\end{proof}

%% file: pd.tex

In this section, we analyze the parameterized complexity of \pdrec{} and \pdman{}. As with \pvrec, it is easy to see that \pdrec{} is also FPT when parameterized by the number of districts ($k$), since we may guess the districts to be recounted. Since $n \geq k$, it is also FPT parameterized by the number of voters. 

\begin{proposition}
    \pdrec{} is FPT when parameterized by the number of districts $k$ or the number of voters $n$.
\end{proposition}


We now show that \pdrec{} is W[1]-hard parameterized by budget of the defender \bd{}. Before describing the construction formally, we briefly outline the main idea. We reduce from the \textsc{Multicolored-Clique  Problem}, which is well-known to be W[1]-hard parameterized by the no. of color classes, which we denote by $k$. Let $(G = (V = V_1 \uplus \cdots \uplus V_k,E), k)$ be an instance of \textsc{Multicolored-Clique}. In the reduced instance, we introduce a special candidate $w$ who is the preferred candidate of the defender. The rival candidates are candidates corresponding to color classes $\mathcal{R}_i$ and ordered pairs of color classes $\mathcal{R}_{ij}$. Further, we introduce candidates encoding the vertices $c_v$ and edges $h_e$ $\&$ $a_e$. We also introduce some \emph{dummy} candidates. We introduce \emph{two} districts corresponding to each $v \in V$ and \emph{five} districts corresponding to each $e\in E$. Also, there exists a baseline district for each candidate which is ``immutable'', which sets up the initial scores of all candidates such that they are equal to a large number (say $\lambda$). The scores are set up in such a way that there is no way to increase the score of $w$, thus we require to reduce the score of all the rivals by at least one while not increasing the scores of other candidates. But the districts and voting profiles are engineered so as to enforce that any recounting solution must have a certain structure, from which we can draw a correspondence to a subset of vertices which must in fact form a multi-colored clique of size $k$ in $G$.


\begin{lemma}
    \pdrec{} is W[1]-hard parameterized by \bd, the defender's  budget.
\end{lemma}
\begin{proof}
 This hardness result follows from the following reduction from the \mulcliq{} problem. The given instance is the graph $G = (V,E)$ and the number of unique color classes, $k$, where $V_i$ denotes the $i^{th}$ color class. We begin by describing the construction of the reduced instance.

\textbf{Candidates:} For every color class $1 \leq i \leq k$ there is a candidate $\mathcal{R}_i$ corresponding to $V_i$. Further, for every pair of color classes $(i,j)$ such that $1 \leq i < j \leq k$, we introduce \emph{two} candidates $\mathcal{R}_{ij}$ and $\mathcal{R}_{ji}$. These will be the rival candidates of the reduced instance. Now we introduce candidates that encode the vertices and edges of the graph $G$. To begin with, for each vertex $v \in V$ we introduce two candidates $c_v$ and $d_v$, which we will refer to as the \emph{main} and \emph{dummy} candidates, respectively. Also, for every edge $e \in E$, we introduce two candidates $h_{e}$ and $a_{e}$, which we refer to as the \emph{helper} and \emph{auxiliary} candidates, respectively. Finally, we have a \emph{special} candidate $w$. To summarize, the overall set of candidates is:

$$C = \{R_i ~|~ i \in [k]\} \cup \left\{R_{ij}, R_{ji} ~|~ (i,j) \in \binom{[k]}{2} \right\} \cup \{c_v,d_v ~|~ v \in V \} \cup \{a_e, h_e ~|~ e \in E\}$$


\textbf{Districts:} We introduce the following districts.

\begin{enumerate}
    \item For each $v \in V$ we introduce a \emph{primary district} $D_v$ with weight one and a \emph{critical district} $D^\star_v$ with weight $k$. 
    \item For each $e = (u,v) \in E$, we introduce two \emph{edge districts} $D_{uv}$ and $D_{vu}$, one \emph{support district} $S_{e}$, and two \emph{transfer districts} $T_{e,u}$ and $T_{e,v}$. The support districts have weight two, while the remaining districts have weight one. 
    \item For each candidate $c \in C \setminus \{d_v ~|~ v \in V\}$, we introduce a \emph{baseline district} $B_c$ with a weight of $\lambda_c$, which will be specified in due course. 
\end{enumerate}


\textbf{Voting outcomes:} The voting outcomes in the original and manipulated districts are depicted in the table below.

\bgroup
\def\arraystretch{1.85}
\setlength{\tabcolsep}{9pt}
\setlength{\LTcapwidth}{\textwidth}
\begin{longtable}[c]{lllllllll}
    \label{tab:votingoutcomes}\\
    District Type      & $D_v$   & $D_v^\star$ & $D_{uv}$           & $D_{vu}$           & $S_e$ & $T_{e,u}$ & $T_{e,v}$ & $B_c$  \\\hline
    \endfirsthead
    \endhead
    Original Winner    & $c_v$   & $d_v$       & $h_e$              & $h_e$              & $a_e$ & $c_u$     & $c_v$   &  $c$   \\\hline
    Manipulated Winner & $R_{j}$ & $c_v$       & $\mathcal{R}_{ij}$ & $\mathcal{R}_{ji}$ & $h_e$ & $a_e$     & $a_e$   &  $c$   \\\hline
    Weight             & $1$     & $k$         & $1$                & $1$                & $2$   & $1$       & $1$     &  $\lambda_c$  \\
    \caption{For the winners depicted above assume that $v \in V_j$ and that $e = (u,v)$, and further that $u \in V_i$ and $v \in V_j$. The subscript $c$ in the last column corresponding to the baseline districts denotes an arbitrary non-dummy candidate.}
\end{longtable}
\egroup

Note that the voting outcome in the baseline districts is the same in the original and manipulated settings. It only remains to specify explicitly the weights of the baseline districts. We let $\lambda_w := (n+1)k + 6m$. Recall that this is the weight of the baseline district corresponding to the special candidate. For any non-dummy candidate $c$, let $s(c)$ be its score from the manipulated districts. We then set $\lambda_c := \lambda_w - s(c)$. With these weights, we ensure that all the non-dummy candidates tie for the same score (i.e, $\lambda_w$) in the manipulated election, and all dummy candidates have a score of zero. We note that all the weights introduced here are polynomially bounded. We set $B_\mathcal{D} = 2k + 5\times \binom{k}{2}$. The preferred candidate is $w$.

We enforce the following tie-breaking order: $\hdots \mathcal{R}_i \hdots \succ \hdots \mathcal{R}_{i j} \hdots \succ w \succ \hdots c_{v} \hdots \succ \hdots a_e \hdots \succ \hdots h_e \hdots \succ \hdots d_v \hdots.$






This completes the description of the constructed instance. We now turn to the proof of equivalence.

\textbf{Forward Direction.} Let a multi-colored clique $V^\star = \{v_1, \ldots, v_k\}$ of size $k$ be given, and without loss of generality, we assume that $v_i \in V_i$. Now, recount in districts $D_v$ and $D^\star_v$ for all $v \in V^\star$. Also, for every edge $e = (u,v)$ in $G[V^\star]$, we recount in the districts $D_{uv}$, $D_{vu}$, $S_e$, $T_{e,u}$ and $T_{e,v}$. 

We claim that this recounting strategy leads to a win for $w$. To begin with,  observe that the score of every rival candidate drops by one. Indeed, for all $i \in [k]$, the score of $\mathcal{R}_i$ reduces by one having lost the primary district $D_{v_i}$. Further, for all $(i,j) \in \binom{[k]}{2}$, the scores of $\mathcal{R}_{ij}$ and $\mathcal{R}_{ji}$ reduce by one each, having lost in the edge districts $D_{v_i v_j}$ and $D_{v_j v_i}$, respectively. It is easy to see that the score of $w$ and any candidate corresponding to vertices and edges \emph{not} involved in the clique remain unchanged. 

Now, consider the auxiliary and helper candidates $a_e, h_e$ corresponding to an edge $e$ in $G[V^\star]$. The score of the helper candidate increases by two in the edge districts used to beat the rival candidates, but also decreases by two because of the recounting in the support district corresponding to $e$. This causes the score of the auxiliary candidate to increase by two, but the recounting the two transfer districts again decreases its score by two. Therefore, in terms of total score, the auxiliary and helper candidates are back to where they started. 

Meanwhile, the recounting across all transfer districts causes the scores of all the main candidates corresponding to vertices of the clique to increase by $(k-1)$, and their score also increases by one in the recounted primary districts used to decrease the score of the rival candidates corresponding to vertices. However, recounting in the critical districts reduces their score by $k$, and we conclude that the net change of score is zero for the main candidates as well. The recounting across the critical districts causes the scores of some dummy candidates to increase by $k < \lambda_w$. The scores of all other dummy candidates remains unchanged at zero. Based on the tie-breaking order, it is easy to verify that $w$ wins under this score profile (all rival candidates ranked ahead score less than $w$ and all other candidates that are tied with $w$ are ranked below $w$ in the tie-breaking order). This concludes the argument in the forward direction. 




\textbf{Reverse Direction.} Conversely, suppose that the defender has a recounting strategy $\mathcal{Z}$ that results in making $w$ the winner. Since the rival candidates dominate $w$ in the tie-breaking order and have the same score as $w$ among the manipulated districts, it is imperative for the defender to force every rival candidate to lose at least one of the districts that it wins in the manipulated setting. In particular, $\mathcal{Z}$ must include at least $k$ primary districts and $2\binom{k}{2}$ edge districts. Let us say that a solution $\mathcal{Z}$ is \emph{well-formed} if it consists of $k$ primary districts, $k$ critical districts, and $2\binom{k}{2}$ edge districts, $\binom{k}{2}$ support districts and $2\binom{k}{2}$ transfer districts, and further, that for any $e = (u,v) \in E$:

$$D_{uv} \in \mathcal{Z} \iff D_{vu} \in \mathcal{Z} \iff S_{e} \in \mathcal{Z} \iff T_{e,u} \in \mathcal{Z} \mbox{ and } T_{e,v} \in \mathcal{Z} .$$

A well-formed solution naturally corresponds to a subset of $\binom{k}{2}$ edges, which we will now refer to as the \emph{affected edges}. Note that any vertex incident to an affected edge emerges as a threat to $w$ because of the recounting in the transfer districts $T_{e,u}$ and $T_{e,v}$, which are won by $c_u$ and $c_v$ respectively in the original election. Now, observe that the only way to ``fix'' this is to recount in a critical district corresponding to the vertex in question. Therefore, if the affected edges span $\ell$ vertices, we are forced to recount at least $\ell$ critical districts. Recall that the rival candidates corresponding to color classes impose a requirement of recounting $k$ primary districts. Thus, given our overall budget of $2k + 5\cdot\binom{k}{2}$, we conclude that the $\binom{k}{2}$ affected edges derived from any well-formed solution must in fact span exactly $k$ vertices, and it is straightforward to verify that these vertices will correspond to a multi-colored clique in $G$. 

From the discussion above, it suffices to show that any valid solution $\mathcal{Z}$ must, in fact, be a well-formed solution. To begin with, note that the rival candidates corresponding to pairs of color classes force a recount in $2\binom{k}{2}$ of the edge districts. We first argue that these must correspond to $\binom{k}{2}$ distinct edges, in other words, $D_{uv} \in \mathcal{Z} \iff D_{vu} \in \mathcal{Z}$. Indeed, suppose not. Consider the set $F \subseteq E$ of all edges corresponding to the recounted edge districts in $\mathcal{Z}$, and note that our assumption implies that $|F| > \binom{k}{2}$. 

Now observe that the recounting in the edge districts turns the corresponding helper edge candidates into rivals. The only way to reverse this damage is to recount in the support districts corresponding to these edges, which are the only districts that are won by the helper candidates in the manipulated election. However, these recounts in turn cause the ``partner'' auxiliary candidates to emerge as new rivals whose scores are \emph{two} more than the score of $w$ --- recall that the support districts had a weight of two. Again, based on the voting outcomes, we see that we are now forced to recount in \emph{both} the transfer districts corresponding to the auxiliary candidates, since these are again the only districts that are won by the auxiliary candidates in question. As a result of the recounting in the transfer districts, the score of every vertex incident to any edge in $F$ has strictly increased. Note, however, that since $|F| > \binom{k}{2}$, $F$ necessarily spans more than $k$ vertices in $G$, each of which are now a threat to $w$. This forces a recount of as many critical districts for $w$ to win, but this contradicts the available budget. Therefore, we conclude that $D_{uv} \in \mathcal{Z} \iff D_{vu} \in \mathcal{Z}$. It is easy to see that the recounting of the corresponding support and transfer districts are forced based on the voting outcomes and the arguments made earlier in this discussion. Thus we conclude that any valid solution is well-formed, which in turn leads to a natural selection of vertices corresponding to a clique in $G$, completing the argument in the reverse direction.
\end{proof}


We now turn our attention to \pdman. We observe that \pdman{} can also be shown to be W[1]-hard by a reduction from \pdrec. We briefly sketch the main idea: to begin with, we switch the roles of the manipulated profiles and the original ones, set the defender budget to zero and the attacker budget to the defender budget. We would also set up the votes in the districts to be such that the only meaningful manipulation by the attacker is to move the manipulated profile to the original one. The equivalence is based on repurposing a recounting strategy to an attacking one and vice-versa. 


%% file: concl.tex
Our main contribution was to settle the parameterized complexity for the problems of recounting and manipulation when parameterized by the player budgets, for both the PD and PV implementations of the plurality voting rules. We also observed that these problems are FPT when parameterized by the number of voters, and that the recounting problem is FPT when parameterized by the number of districts as well. 

We make some remarks about directions for future work. In the setting of succinct input, the problems of recounting and manipulation are already para-NP-hard because of the NP-completeness for three candidates. When the votes are specified in unary is an interesting direction for future work. The dynamic programming algorithm proposed by~\cite{ElkindGORV19} already shows that the problem is in XP, parameterized by the number of candidates, and we leave open the issue of whether the problem is FPT. The problem of manipulation parameterized by the number of districts is another unresolved case. More broadly, it would be interesting to challenge the theoretical hardness results obtained here against heuristics employed on real world data sets. The issue of identifying and working with structural parameters is also an interesting direction for further thought. 
